\documentclass[reqno,11pt]{amsart}
\usepackage{enumerate}
\usepackage{color}
\usepackage{amssymb}
\usepackage{amsbsy}
\usepackage{amsmath}
\usepackage{bbm}
\usepackage{float}
\usepackage{comment}
\usepackage{pgf}

\allowdisplaybreaks

\newtheorem{theorem}{Theorem}[section]
\newtheorem{lemma}[theorem]{Lemma}
\newtheorem{proposition}[theorem]{Proposition}
\newtheorem{corollary}[theorem]{Corollary}

\theoremstyle{definition}
\newtheorem{definition}[theorem]{Definition}

\theoremstyle{remark}

\numberwithin{equation}{section}

\begin{document}

\title{On  Simultaneous Percolation with Two Disk Types}

\author{Michal Yemini}
\address{Faculty of Engineering, Bar-Ilan University,
Ramat Gan 52900, Israel}
\email{michal.yemini.biu@gmail.com}

\author{Anelia Somekh-Baruch}
\address{Faculty of Engineering, Bar-Ilan University,
Ramat Gan 52900, Israel}
\email{anelia.somekhbaruch@gmail.com}

\author{Reuven Cohen}
\address{Department of Mathematics, Bar-Ilan University,
Ramat Gan 52900, Israel}
\email{reuven@math.biu.ac.il}

\author{Amir Leshem}
\address{Faculty of Engineering, Bar-Ilan University,
Ramat Gan 52900, Israel}
\email{leshem.amir2@gmail.com}
\thanks{This research was partially supported by Israel Science Foundation (ISF) grant 903/2013.}

\date{\today}


\begin{abstract}
In this paper we consider the simultaneous percolation of two Gilbert disk models. The two models are connected through excluding disks, which prevent elements of the second model to be in the vicinity of the first model. Under these assumptions we characterize the region of densities in which the two models both have a unique infinite connected component.
The motivation for this work is  the co-existence of two cognitive radio networks.\end{abstract}

\maketitle

\section{Introduction}

One of the most desirable  properties of networks in communication theory is connectivity.
This property makes it possible to pass information between nodes from all over the network. Connectivity in large and immobile  networks can be verified by simple  routing algorithms such as Dijkstra's algorithm. However, the verification of the connectivity of mobile ad-hoc networks is less obvious since the node locations  vary over time.

One method for analyzing the connectivity of mobile ad-hoc networks is continuum percolation,  for example see \cite{DousseBaccelli2005,DousseFranceschetti2006,RenZhaoSwami2011,RenZhao2011,LuHuang2012,LiuLiu2015}. Assume that  the nodes of the network are distributed according to a Poisson point process (PPP) with density $\lambda$. Under a Gaussian channel assumption, the capacity of a link between two nodes is a decreasing function of the distance between the nodes.
Therefore, one can choose a distance $\rho$ such that if two nodes are within a distance $\rho$  they can communicate, or are said to be connected with one another.
This condition is equivalent to the condition that the two disks, each centered at the nodes' positions and radii $\frac{\rho}{2}$, intersect.
In this way we model a homogeneous network by a Gilbert disk (Boolean) model with density $\lambda$ and a fixed radius $\frac{\rho}{2}$.
A network is said to be connected (or percolated) under continuum percolation models if there exists an unbounded connected component in the network. It follows by \cite[Corollary 4.1]{MeesterRoy1996} that  there exist top to bottom ($T-B$) and left to right ($L-R$) crossings (see Section \ref{subsec:crossings} and \cite[chapter 4]{MeesterRoy1996}) in the Boolean model if and only if the model is percolated a.s.\

\begin{figure}
  \centering
  \includegraphics[scale=0.35]{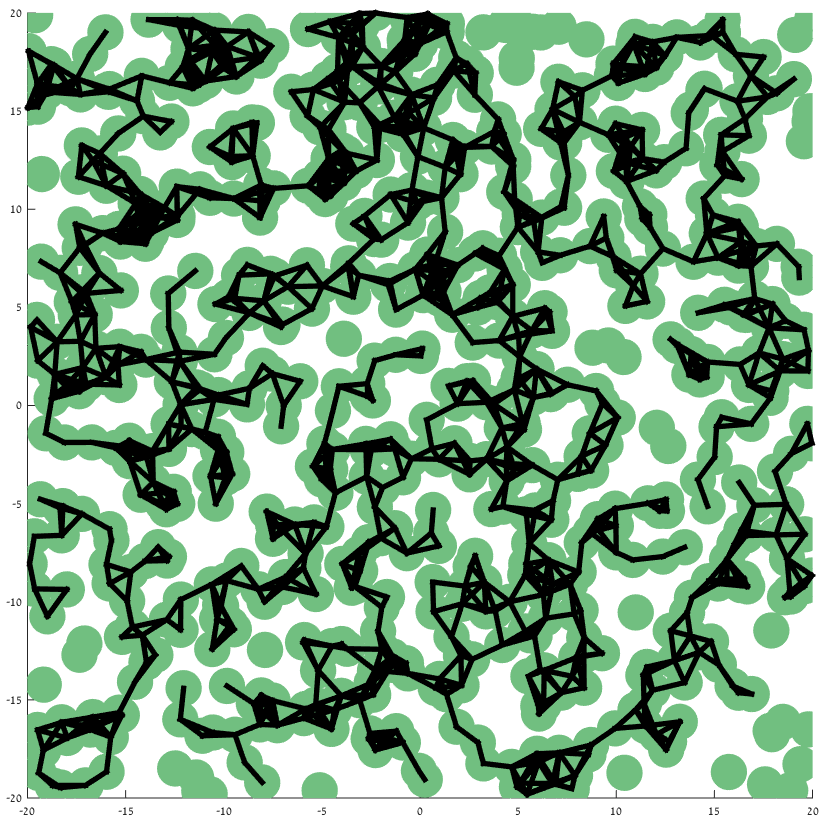}
    \caption{A realization of a homogeneous network with the parameters set as $\lambda=50\text{ km}^{-2}$ and  $\rho=D_t=180\text{ m}$. The black lines indicate the disks that compose the largest connected component. }\label{fig_primary_set1}
  \end{figure}

 The connectivity of large-scale homogeneous networks has been studied extensively in many papers, including \cite{GuptaKumar2000,Philips1989,Bettstetter2002,NiChandler1994,DousseBaccelli2005,DousseFranceschetti2006}. All of these models consider one network scenario. That is, it is assumed that all the nodes in the area of interest belong to the same network.
A realization of such a network is depicted in Fig.\ \ref{fig_primary_set1}.

In recent years, due to the scarcity of free static spectrum resources, a new concept dubbed  “Cognitive Radio” has emerged. The overarching goal of cognitive radio networks is to improve spectrum utilization by giving communication opportunities to cognitive (secondary) nodes  while limiting their interference on non-cognitive (primary) nodes in the network.
Several papers such as \cite{RenZhao2011,Ao2012,LuHuang2012,Tameemi2014,LiuLiu2015} have considered the connectivity of non-cooperative networks with heterogeneous nodes. Although these articles consider systems with primary and secondary networks, they are restricted to the connectivity of the secondary network. This leads to models in which the secondary nodes are components of a multi-hop network, but the primary nodes are only a part of a single-hop network. These degenerate primary networks do not capture the connectivity demands of the multi-hop primary network;  hence our motivation to analyze a generalized heterogeneous model in which both the primary and secondary nodes compose multi-hop networks.

In this paper we pursue the connectivity of both the secondary network and the primary network, a state we call simultaneous connectivity.  We assume that the nodes of the primary network are distributed according to a two-dimensional PPP. Additionally, the secondary nodes are distributed according to a two-dimensional PPP which is independent of the primary network's PPP. Like the homogeneous model, the capacity of a link  under a Gaussian channel assumption is a decreasing function of the link length and an increasing function of the closest interfering node that belongs to the other network. Consequently, we  assume that each primary node has a transmission range $D_t$, and every secondary node has a transmission range $d_t$. Also, we assume that a secondary node is not active, unless its distance to each primary node is greater than $D_f$.
That is, $D_f$ is the radius of the guard zone of primary nodes from the interference of secondary nodes and vice versa.
We call this model the \textit{heterogeneous model}. A realization of such a model is depicted in Figs.\ \ref{fig_primary_set1}-\ref{fig_secondary_set1}
where Fig.\ \ref{fig_primary_set1} depicts the primary network, Fig.\ \ref{fig_primary_secondary_set1} includes the guard zone of
each primary user and active and passive secondary nodes, and Fig.\ \ref{fig_secondary_set1} depicts the active nodes of the secondary network.
\begin{figure}
  \centering
  \includegraphics[scale=0.5]{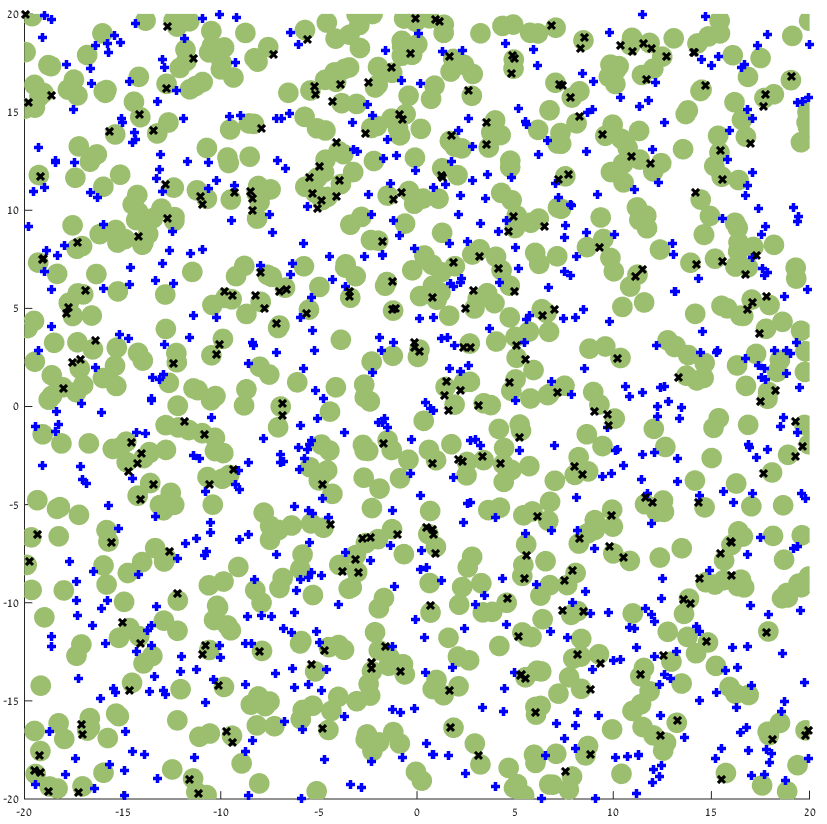}
    \caption{The guard zone of the primary nodes which appear in Fig.~\ref{fig_primary_set1} and a realization of the active and passive secondary nodes. We set $\lambda_s=\lambda_p=50\text{km}^{-2}$, $D_f=50 \text{ m}$. Green disks depict the guard zone of the primary nodes, blue $+$ indicates active secondary nodes and black x indicates passive secondary nodes.}\label{fig_primary_secondary_set1}
 \end{figure}
\begin{figure}
 \centering
  \includegraphics[scale=0.35]{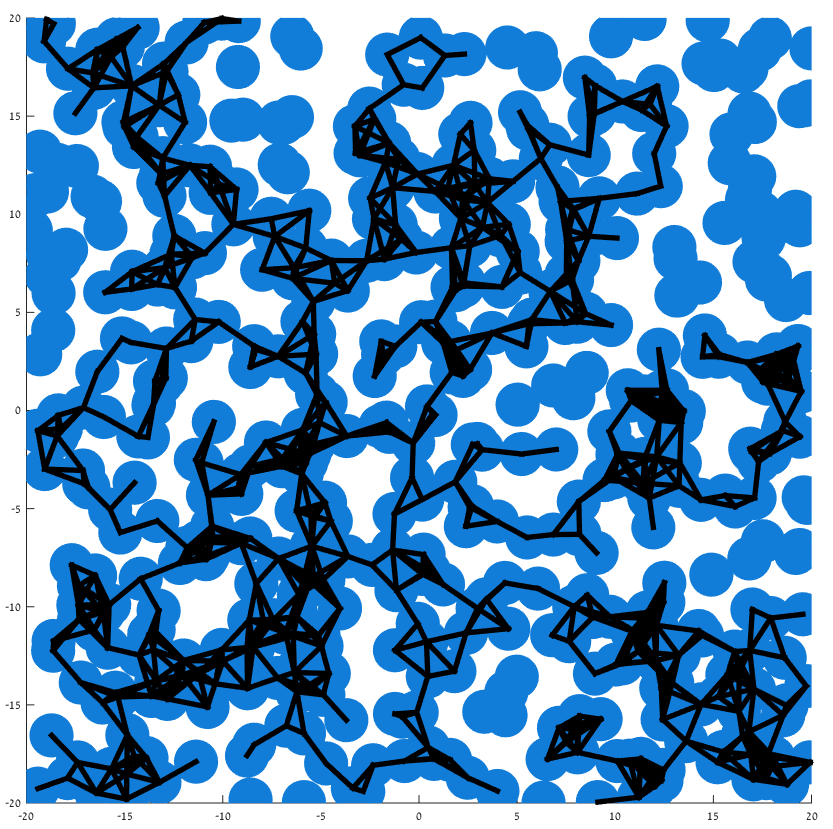}
    \caption{The active secondary nodes of the realization which appears in Fig.~\ref{fig_primary_secondary_set1} and the parameter $d_t=220\text{ m}$. The black lines indicate the disks that compose the largest connected component of the secondary network.}\label{fig_secondary_set1}
  \end{figure}

Define the simultaneous connectivity region to be all density pairs for which there is at least one unbounded connected component in  the primary network as well as an unbounded connected component in the secondary network. Our aim is to analyze the relationship between the densities of the two dimensional PPPs driving the primary and secondary networks and the radii $D_t,d_t$ and $D_f$.
 The analysis conducted in this paper aims to  provide motivations and insights into generalizations of several current applications of stochastic geometry such as routing \cite{1512115,7088595,1522131,6600702}, medium-access control \cite{BaccelliBlaszczyszyn2006,5226963,5457860,6354273,6583304} and interference analysis  in  wireless networks \cite{DousseFranceschettiMacris2006,4453888,HaenggiGantiNOW2009,6155562,HeathKountouris2013}. Since there are many works on these topics we only cite the several works and the references therein. For further reading regarding applications of stochastic geometry in the analysis and design of wireless ad-hoc networks see for example   \cite{HaenggiGantiNOW2009,BaccelliBlaszczyszynNOW20091,BaccelliBlaszczyszynNOW20092,HaenggiAndrewsBaccelli2009,Haenggi2013}.

The rest of this paper is organized as follows. Section \ref{sec:not_def}  presents the heterogeneous network model. We also state definitions and results from Percolation Theory. Section \ref{sec:Connectedness} establishes the connectedness of the simultaneous connectivity region.  In Section \ref{sec:con_comp} we claim and prove that the heterogeneous model is ergodic and  that there cannot exist more than one unbounded connected component in each of the networks. Section \ref{sec:con_comp}
shows that for  densities  greater than the critical densities of the two networks when $D_f=0$ there exists $D_f>0$ such that both networks still contain an unbounded connected component. Section \ref{sec:nec_cond} presents the necessary conditions for  dual connectivity.   Section \ref{sec:suf_cond} covers the sufficient conditions for the simultaneous connectivity of the heterogeneous model. Section \ref{sec:conclusion} concludes the paper.

\section{System Model and Definitions}\label{sec:not_def}

 In this section we present the heterogeneous model. We also state fundamental definitions and results of Percolation Theory which we apply in our analysis of the connectivity of the heterogeneous model.
\subsection{The Heterogeneous Model}
 In this model the primary nodes are distributed  according to a two-dimensional PPP with  density $\lambda_p$.
We assume that the transmission range of primary nodes , $D_t$, is fixed.
 Similarly, the nodes of the secondary network are distributed according to a PPP with density $\lambda_s$ .This PPP is independent of the PPP of the primary network. We also assume that the transmission range of secondary nodes, $d_t$, is fixed. We next provide several definitions  corresponding to the heterogeneous model.

\begin{definition}[Communication Opportunities of Primary Nodes]\label{secon:com_opp}
There is a communication opportunity from node $x_i$ to node $x_j$ in the primary network if
  $\|x_i-x_j\|_{2}\leq D_t$, where $\|\cdot\|_2$ denotes the $L_2$ norm.
 \end{definition}
 \begin{definition}[Communication Opportunities of Secondary Nodes]
There is a communication opportunity from node $z_i$ to node $z_j$  in the secondary network if the following conditions hold:
  \begin{enumerate}
     \item   $\|z_i-z_j\|_{2}\leq d_t$,
     \item there is no primary node $x$ such that $\|x-z_i\|_{2}\leq D_f$,
     \item there is no primary node $x$ such that $\|x-z_j\|_{2}\leq D_f$.
      \end{enumerate}
 \end{definition}
 Consequently,  $D_f$ is the radius of the guard zone of a primary node and  transmitting/receiving secondary nodes.

We only discuss bidirectional links; that is, we say that there is a link between the nodes $z_i$ and $z_j$  if there exists a
communication opportunity from node $z_i$ to node $z_j$ and vice versa.

 \begin{definition}
Let $X_p$ be the set of nodes of the primary network. The connected component of node $x\in X_p$  consists of all nodes in $X_p$ for which there exists a path to $x$  in  $X_p$  such that every two consecutive nodes in the path have a communication opportunity.
Additionally,  an unbounded connected component of the primary network is a connected component
of the primary network which consists of an infinite number of nodes.

The definition of a connected component in the secondary network is similar. \end{definition}

 \begin{definition}
 We define the following simultaneous connectivity regions:
 \begin{itemize}
 \item The simultaneous connectivity region $\mathcal{C}(D_t,d_t,D_f)$  consists of all pairs of densities $(\lambda_p,\lambda_s)$ such that both the primary and secondary networks include a.s.\ at least one unbounded connected component for a given vector of parameters $(D_t,d_t,D_f)$.
\item The simultaneous connectivity region  $\mathcal{C}(D_t,d_t)$
consists of all triples   $(D_{f},\lambda_p,\lambda_s)$ such that both the primary and secondary networks include a.s.\ at least one unbounded connected component for a given vector of transmission radii $(D_t,d_t)$.
 \item The simultaneous connectivity region  $\mathcal{C}(D_t)$
consists of all $4$-tuples $(d_t,D_f,\lambda_p,\lambda_s)$ such that both the primary and secondary networks include a.s.\ at least one unbounded connected component for a given vector of transmission radius $D_t$.
 \item The simultaneous connectivity region  $\mathcal{C}$
consists of all $5$-tuples $(D_{t},d_t,D_f,\lambda_p,\lambda_s)$ such that both the primary and secondary networks include a.s.\ at least one unbounded connected component.
 \end{itemize}
 \end{definition}

 The connectivity of the primary and secondary networks can be studied by representing the two networks by the two independent Boolean models (see the following Section). Nevertheless, our connectivity definitions differ from those of a simple Boolean model (see \cite{MeesterRoy1996}).

 We now provide some definitions which are required for the analysis of the heterogeneous model.
 \subsection{The Gilbert Disk (Boolean) Model}\label{appendix:Boolean_model}
 The Gilbert disk model scatters points in $\mathbb{R}^2$ according to a PPP. Each point in the PPP is assumed to have a fixed radius. In the following we present several definitions related to  the Gilbert disk model.
 \begin{definition}[Point Process]
 Let $\mathcal{B}^2$ be the $\sigma$-algebra of Borel sets in $\mathbb{R}^2$, and let $N$ be the set of all simple counting measures on $\mathcal{B}^2$. Let $\mathcal{N}$  be the $\sigma$-algebra which is generated by the sets
\begin{flalign}
\{n\in N:n(A)=k\},
\end{flalign}
where $A\in\mathcal{B}^2$, and $k$ is an integer. A point process $X$ is  a measurable mapping from a probability space $(\Omega,\mathcal{F},P)$ into $(N,\mathcal{N})$. The distribution of $X$ is denoted by $\mu$ and is defined by $\mu(G)=P(X^{-1}(G))$, for all $G\in\mathcal{N}$. Hereafter, for convenience we refer to $(N,\mathcal{N})$ as $(\Omega,\mathcal{F})$ .
 \end{definition}

\begin{definition}[Gilbert Disk (Boolean) Model]
Suppose that $X$ is a point process. A Gilbert disk (Boolean) model is composed of point process $X$ and a fixed radius $\rho$ such that  each point $x\in X$ is a center of a disk with a fixed radius $\rho$.
\end{definition}
Note that this model  is equivalent to a Boolean model with fixed radii. In this paper we assume that $X$ is a PPP with density $\lambda.$
We denote this Poisson  Gilbert disk (Boolean) model   by $(X,\rho,\lambda)$.

We represent the heterogeneous network by the following Gilbert disk models $(X_p,D_t/2,\lambda_p)$ and $(X_s,d_{t}/2,\lambda_s)$, where $(\Omega_p,\mathcal{F}_p,P_p)$ and $(\Omega_s,\mathcal{F}_s,P_s)$  are the probability spaces of two independent PPPs $X_p$ and $X_s$, respectively. Further, $D_t$ and $d_t$ are the transmission radii in the primary and secondary networks, respectively.

\subsection{Occupied Components}Define $O(z,r)\triangleq\{x\in\mathbb{R}^2:\|x-z\|_2\leq r\}$.

Every Poisson Boolean model $(X,\rho,\lambda)$ partitions $\mathbb{R}^2$ into  two regions, the \textit{occupied region}, which we denote by
\begin{flalign}
\mathcal{O}\triangleq\bigcup_{x\in X}O(x,\rho),
\end{flalign}
and the \textit{vacant region}. The occupied region consists of the points in $\mathbb{R}^2$ that are covered by at least one disk, whereas the vacant region consists of all points in $\mathbb{R}^2$ that are not covered by any disk.

Two nodes $x_1,x_2\in X$ are connected if $O(x_1,\rho)\cap O(x_2,\rho)\neq\emptyset$ (however, in the secondary network we  consider only active nodes). The  connected components in the occupied region are called \emph{occupied components}, whereas the connected components in the vacant region are called \emph{vacant components}. Additionally, for $A\subset \mathbb{R}^2$
\begin{flalign}
W(A) \triangleq &\text{ the union of all occupied components which have}\nonumber\\ &\hspace{1.5cm}\text{a non-empty intersection with }A.
\end{flalign}
 When $A=\{0\}$, we write $W\triangleq W(0)$ and we call $W$ the occupied component of the origin. We use similar definitions for the vacant components which we denote by $V$. Note that only one of the components $V$ and $W$ can be empty. 

 Note that by  definition of the occupied components in the Boolean model, two nodes are connected if the distance between them does not exceed $2\rho$. Therefore, we represent each network by a model in which $\rho$ is half of the transmission radius.
Further, an occupied component/region in the secondary network consists solely of  active secondary nodes of the secondary network.
\subsection{The Critical Probability}
We next define the critical probability of the Gilbert disk model.\begin{definition}[Critical Probability]
Let $d(A)\triangleq\sup_{x,y\in A}|x-y|$ . Denote by $\theta_{\rho}(\lambda)$ the probability that the origin is an element of an unbounded occupied component of the Gilbert disk (Boolean) model $(X,\rho,\lambda)$, that is
\[\theta_{\rho}(\lambda)\triangleq\Pr(d(W)=\infty).\]
The critical density $\lambda_c(2\rho)$ is defined by
\begin{flalign}
\lambda_c(2\rho)\triangleq \inf\{\lambda\geq0: \theta_{\rho}(\lambda)>0\}.
\end{flalign}
\end{definition}
As we next state, the critical probability has a strong tie to the crossing probabilities which we define next.
\subsection{Crossings Probabilities}\label{subsec:crossings}
A continuous path is said to be \textit{occupied} if it lies in an occupied component. An occupied path $\upsilon$ is an occupied $L-R$ crossing of the rectangle $[0,l_1]\times[0,l_2]$ if there exists a segment $\tau$ of $\upsilon$  which is contained in the rectangle $[0,l_1]\times[0,l_2]$ and it also intersects the left and right boundaries of the rectangle. We define an occupied $T-B$ crossing of the rectangle in a similar manner.

Let $\sigma((l_{1},l_{2}),\lambda,L-R)$ be the $L-R$ crossing probability  of the rectangle $[0,l_{1}]\times[0,l_{2}]$, that is, the probability of the existence of an occupied $L-R$ crossing.
Also, let  $\sigma((l_{1},l_{2}),\lambda,T-B)$ denote the $T-B$ crossing probability  of the rectangle $[0,l_{1}]\times[0,l_{2}]$. Suppose that $(X,\lambda,\rho)$ is a two-dimensional Gilbert disk (Boolean)  model with a bounded $\rho$. Then by \cite[Corollary 4.1]{MeesterRoy1996} it follows that for all $k\geq 1$,
\begin{flalign}
\lim_{n\rightarrow\infty}\sigma((kn,n),\lambda,L-R)=\begin{cases}
1, & \text{ if } \lambda>\lambda_c(2\rho)\\
0, & \text{ if } \lambda<\lambda_c(2\rho)
\end{cases}.
\end{flalign}
By symmetry, a similar result holds for the $T-B$ crossings, that is,
\begin{flalign}
\lim_{n\rightarrow\infty}\sigma((n,kn),\lambda,T-B)=\begin{cases}
1, & \text{ if } \lambda>\lambda_c(2\rho)\\
0, & \text{ if } \lambda<\lambda_c(2\rho)
\end{cases}.
\end{flalign}

\subsection{Unit Transformations}
We now define unit transformations of the heterogeneous model. We later use this definition in the discussion of the ergodicity of the heterogeneous model. \begin{definition}
Let  $\mathcal{B}^2$ denote the Borel sets of $\mathbb{R}^2$ and let $\Omega$ be a set of simple counting measures on $\mathcal{B}^2$.
Let $t\in\mathbb{R}^2$ and $T_t:\mathbb{R}^2\rightarrow\mathbb{R}^2$ be defined by the translation $T_{t}x=x+t$. $T_t$ then induces the transformation $S_{t}:\Omega\rightarrow\Omega$ for each $A\in\mathcal{B}^2$ through the equation\footnote{See \cite{MeesterRoy1996} page 22.}
\begin{flalign}\label{def:S_ei}
(S_{t}\omega)(A)=\omega(T_{t}^{-1}A),\: \forall\omega\in\Omega.
\end{flalign}

\end{definition}
Let $\tilde{\Omega}=\Omega_p\times \Omega_s$, $\tilde{\mathcal{F}}=\mathcal{F}_p\times \mathcal{F}_s$ and $\tilde{P}=P_p\times P_s$.
Denote the unit vectors of $\mathbb{R}^2$ by $e_1,e_{2}$.
It follows that $T_{t}$ induces the transformation $\tilde{T}_{t}$ on $\tilde{\Omega}$ where
\begin{flalign}\label{def:t_t}
\tilde{T}_{t}=(S_{t}\omega_p,S_{t}\omega_s).
\end{flalign}
More specifically, $T_{e_{i}}$ induces the transformation $\tilde{T}_{e_i}$ on $\tilde{\Omega}$ where
\begin{flalign}\label{def:t_ei}
\tilde{T}_{e_i}=(S_{e_i}\omega_p,S_{e_i}\omega_s).
\end{flalign}

\section{The Connectedness of the Simultaneous Connectivity Region}\label{sec:Connectedness}

In this section we establish the connectedness of the simultaneous connectivity region.

\begin{proposition}\label{prop_connect}
 The  simultaneous connectivity region $\mathcal{C}$ is connected.
 \end{proposition}

 The proof of this theorem relies on the following lemmas.

 \begin{lemma}\label{lemma:connect1}
 The  simultaneous connectivity region $\mathcal{C}(D_t,d_t,D_f)$ is connected for each vector of parameters $(D_t,d_t,D_f)$.
 \end{lemma}
 \begin{proof}

The proof generalizes  the proof in  \cite[ Theorem 1]{RenZhao2011}. Note that as unlike in \cite{RenZhao2011}, we need to ensure the connectivity of both the primary and the secondary networks. For ease of notation we refer to $\mathcal{C}(D_t,d_t,D_f)$ as $\tilde{\mathcal{C}}$.

Let $(\lambda_{p1},\lambda_{s1})$ and $(\lambda_{p2},\lambda_{s2})$ be two pairs of densities in the simultaneous connectivity region.
We assume without loss of generality that $\lambda_{s_1}\leq \lambda_{s_2}$, and prove that there is a  path from $(\lambda_{p1},\lambda_{s1})$ to $(\lambda_{p2},\lambda_{s2})$ which resides in $\tilde{\mathcal{C}}$. We consider the path that is constructed by the horizonal segment which starts at
$(\lambda_{p1},\lambda_{s1})$ and ends at $(\lambda_{p1},\lambda_{s2})$ and the vertical segment which starts at
$(\lambda_{p1},\lambda_{s2})$ and ends at $(\lambda_{p2},\lambda_{s2})$. We now distinguish between two cases:
  case $(a)$ in which $\lambda_{p1}\leq\lambda_{p2}$, and case $(b)$ in which $\lambda_{p1}\geq\lambda_{p2}$.
We present the proof for case $(b)$; case $(a)$ follows similarly.
\begin{figure}
\centering
\begin{minipage}[b]{0.45\linewidth}
\centering
\includegraphics{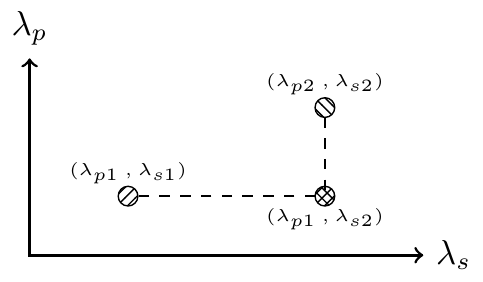}\\(a) $\lambda_{p1}\leq\lambda_{p2}$
\label{fig_crossing}
\end{minipage}
\quad
\begin{minipage}[b]{0.45\linewidth}
\centering
\includegraphics{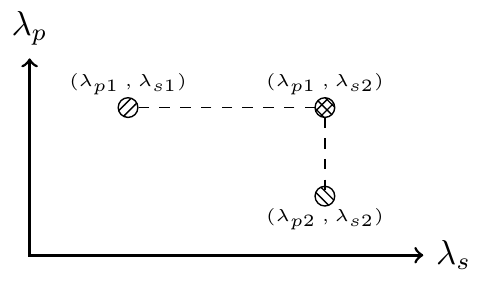}\\(b) $\lambda_{p1}\geq\lambda_{p2}$
\label{fig_crossing}
\end{minipage}
\end{figure}

As mentioned above, we choose the path that consists of two segments. We now prove that each of these segments lies in the simultaneous connectivity region $\tilde{\mathcal{C}}$. First, we show that the segment $(\lambda_{p1},\lambda_s)$ where $\lambda_{s1}\leq\lambda_s\leq \lambda_{s2}$ lies in $\tilde{\mathcal{C}}$. Since the secondary nodes transmit only if they do not interfere with the transmissions of primary nodes, and since  the pair $(\lambda_{p1},\lambda_{s1})$ lies in $\tilde{\mathcal{C}}$, the primary network includes a.s.\ an unbounded connected component for each pair of densities in the segment $(\lambda_{p1},\lambda_s)$. Further since $\lambda_{s1}\leq\lambda_s$, it follows by superposition techniques (see \cite[Theorem 1]{RenZhao2011} and \cite[p. 11]{MeesterRoy1996}) that the secondary network has an unbounded connected component as well. Therefore, the segment  $(\lambda_{p1},\lambda_s)$ where $\lambda_{s1}\leq\lambda_s\leq \lambda_{s2}$
is in $\tilde{\mathcal{C}}$.

Second, we  show that the segment $(\lambda_p,\lambda_{s2})$ where $\lambda_{p2}\leq\lambda_p\leq \lambda_{p1}$ lies in $\tilde{\mathcal{C}}$. One can argue by superposition techniques that for  each of the densities  $(\lambda_{p},\lambda_{s2})$ such that $\lambda_{p2}\leq\lambda_p\leq \lambda_{p1}$, the primary network includes an unbounded connected component a.s.\  Consequently, by  definition \ref{secon:com_opp} of secondary node communication opportunities, if  $(\lambda_{p1},\lambda_{s2})$  lies in $\tilde{\mathcal{C}}$, it follows that  $(\lambda_p,\lambda_{s2})$ lies in $\tilde{\mathcal{C}}$.
\end{proof}

\begin{lemma}\label{lemma_connect2}
The following statements hold for the heterogeneous model:
\begin{enumerate}[(i)]
\item The  simultaneous connectivity region $\mathcal{C}(D_t,d_t)$ is connected for each vector of parameters $(D_t,d_t)$.
\item The  simultaneous connectivity region $\mathcal{C}(D_t)$ is connected for each value of the  parameter $D_t$.
\end{enumerate}
\end{lemma}

\begin{proof}
$ $
\begin{enumerate}[(i)]
\item Let $(D_{f1},\lambda_{p1},\lambda_{s1})$ and $(D_{f2},\lambda_{p2},\lambda_{s2})$ be in $\mathcal{C}(D_t,d_t)$. We show that there exists a path in $\mathcal{C}(D_t,d_t)$ which connects these two triples. Suppose without loss of generality that $D_{f1}\geq D_{f2}$. We prove that $(D_{f},\lambda_{p1},\lambda_{s1})\in \mathcal{C}(D_t,d_t)$ for each $D_{f2}\leq D_f\leq D_{f1}$.
Let $(\omega_{p1},\omega_{s1})$ be a realization  of the heterogeneous  model with a pair of densities $(\lambda_{p1},\lambda_{s1})$ and a guard zone radius $D_{f1}$  such that there is at least one unbounded connected component  in both the primary and secondary networks. For each  such  realization, decreasing the guard zone allows more secondary nodes to be active, or at minimum, all the previously active secondary nodes are still active. Additionally,  decreasing the guard zone does not affect the primary network. It follows that decreasing the guard zone  radius does not reduce the number of unbounded connected components in each of the networks. Therefore, if $(D_{f1},\lambda_{p1},\lambda_{s1})\in \mathcal{C}(D_t,d_t)$ then $(D_{f},\lambda_{p1},\lambda_{s1})\in \mathcal{C}(D_t,d_t)$ for each $D_{f2}\leq D_f\leq D_{f1}$. Since $(D_{f2},\lambda_{p1},\lambda_{s1})$ and $(D_{f2},\lambda_{p2},\lambda_{s2})$ are in $\mathcal{C}(D_t,d_t)$, by Lemma \ref{lemma:connect1} the two triples are connected in $\mathcal{C}(D_t,d_t)$.
\item Let $(d_{t1},D_{f1},\lambda_{p1},\lambda_{s1})$ and $(d_{t2},D_{f2},\lambda_{p2},\lambda_{s2})$ be in $\mathcal{C}(D_t)$. We show that there exists a path in $\mathcal{C}(D_t)$ that connects these two $4$-tuples. Suppose without loss of generality that $d_{t1}\leq d_{t2}$. We prove that $(d_t,D_{f1},\lambda_{p1},\lambda_{s1})\in \mathcal{C}(D_t)$ for each $d_{t1}\leq d_t\leq d_{t2}$.
Let $(\omega_{p1},\omega_{s1})$ be a realization  of the heterogeneous  model with a pair of densities $(\lambda_{p1},\lambda_{s1})$, a guard zone radius $D_{f1}$ and a transmission radius $d_t$  such that there is at least one unbounded connected component  in both the primary and secondary networks. For  this realization,  when increasing the transmission radius of the secondary nodes, all the previously connected nodes are still connected, and we can only connect more nodes in the secondary network without affecting the primary one. It follows that increasing the transmission radius of the secondary nodes does not reduce the number of unbounded connected components in each of the networks. Therefore, if $(d_{t1},D_{f1},\lambda_{p1},\lambda_{s1})\in \mathcal{C}(D_t)$ then $(d_{t},D_{f1},\lambda_{p1},\lambda_{s1})\in \mathcal{C}(D_t)$ for each $d_{t1}\leq d_t\leq d_{t2}$. Since $(d_{t2},D_{f1},\lambda_{p1},\lambda_{s1})$ and $(d_{t2},D_{f2},\lambda_{p2},\lambda_{s2})$ are in $\mathcal{C}(D_t)$, by part (i) of this lemma the two $4$-tuples are connected in $\mathcal{C}(D_t)$.
\end{enumerate}
\end{proof}
\begin{proof}[Proof of Proposition \ref{prop_connect}]
Let $(D_{t1},d_{t1},D_{f1},\lambda_{p1},\lambda_{s1})$ and $(D_{t2},d_{t2},D_{f2},\lambda_{p2},\lambda_{s2})$ be in $\mathcal{C}$. We show that there exists a path in $\mathcal{C}$ that connects these two $5$-tuples. Suppose without loss of generality that $D_{t1}\leq D_{t2}$. We can prove similarly to the proof of the second part of  Lemma \ref{lemma_connect2} that $(D_t,d_{t1},D_{f1},\lambda_{p1},\lambda_{s1})\in \mathcal{C}$ for each $D_{t1}\leq D_t\leq D_{t2}$.  Therefore, if $(D_{t1},d_{t1},D_{f1},\lambda_{p1},\lambda_{s1})\in \mathcal{C}$ then $(D_{t},d_{t1},D_{f1},\lambda_{p1},\lambda_{s1})\in \mathcal{C}$ for each $D_{t1}\leq D_t\leq D_{t2}$. Since $(D_{t2},d_{t1},D_{f1},\lambda_{p1},\lambda_{s1})$ and $(D_{t2},d_{t2},D_{f2},\lambda_{p2},\lambda_{s2})$ are in $\mathcal{C}$, by part (ii) of Lemma \ref{lemma_connect2} the two $5$-tuples\ are connected in $\mathcal{C}$.
\end{proof}

\section{The Unbounded Connected Components}\label{sec:con_comp}
In this section we reach several results regarding the number of unbounded connected components in the primary and secondary networks. We also prove that we can always find $D_f>0$ such that both the primary and the secondary networks include an unbounded connected component a.s.\

In order to prove that there is at most one unbounded connected component in the primary and also in the  secondary networks, we first prove that our heterogeneous model is ergodic.
\begin{proposition}\label{theorem:ergodicity}
The heterogeneous model is ergodic.
\end{proposition}
\begin{proof}
Let $(\Omega_p,\mathcal{F}_p,P_p)$ and $(\Omega_s,\mathcal{F}_s,P_s)$ be the probability spaces of the primary and secondary networks, respectively.  Define the probability space of the heterogeneous network by  $(\tilde{\Omega},\tilde{\mathcal{F}},\tilde{P})$. Further, let
\begin{flalign}
\tilde{T}(\tilde{\omega})_{e_i}=(S_{e_i}\omega_p,S_{e_i}\omega_s),
\end{flalign}
where $S_{e_i}$ is defined in  (\ref{def:S_ei}).

 Since our statistical model consists of two degenerate Boolean models  we can base our proof on results pertaining to  the Boolean model. It follows from the proof of \cite[proposition 2.6]{MeesterRoy1996} that $(\Omega_p,\mathcal{F}_p,P_s,S_{e_i})$ is ergodic. Additionally, it follows from
the proof of \cite[proposition 2.6]{MeesterRoy1996} that $(\Omega_s,\mathcal{F}_s,P_s,S_{e_i})$ is mixing. Therefore, by  \cite[Theorem 6.1]{Petersen1983}, it follows that $\tilde{T}(\tilde{\omega})_{e_i}$ acts ergodically on $(\tilde{\Omega},\tilde{\mathcal{F}},\tilde{P})$. Let $A\in\tilde{\mathcal{F}}$ be an event which is invariant under all transformations $\{\tilde{T}(\tilde{\omega})_z:z\in\mathbb{Z}^2\}$. By definition, $A$ is also invariant under $\tilde{T}(\omega)_{e_i}$, therefore $P(A)=0$ or $P(A)=1$. Thus, $\{\tilde{T}(\tilde{\omega})_z:z\in\mathbb{Z}^2\}$ acts ergodically on  $(\tilde{\Omega},\tilde{\mathcal{F}},\tilde{P})$.
\end{proof}
From the ergodicity of the model we deduce that the number of unbounded connected components in each of the networks is constant a.s.\
\begin{proposition}\label{theorem:number_unbounded_constant}
The number of unbounded connected components in the primary network and the number of unbounded connected components in the secondary network are constant a.s.\
\end{proposition}
\begin{proof}
We prove this proposition by a generalization of the proof in \cite[Theorem 2.1]{MeesterRoy1996}.
Let $N_p,N_s$ be the random number of unbounded  connected components of the primary and secondary networks, respectively. The event
$\{N_p=k_p,N_s=k_s\}$ is invariant under the group  $\{\tilde{T}(\tilde{\omega})_z:z\in\mathbb{Z}^2\}$, for all $k_p,k_s\geq 0$. By the ergodicity of the heterogeneous network, the event  $\{N_p=k_p,N_s=k_s\}$ is of probability $0$ or $1$. Therefore, $N_p,N_s$ are constant a.s.\
\end{proof}

We next prove that there exists at most one unbounded connected component in each network a.s.\

\begin{theorem}\label{theorem:one_unbounded_constant}
There is at most one unbounded connected component in the primary network and at most one unbounded connected component in the secondary network.
\end{theorem}

\begin{proof}
Since the secondary nodes transmit only if they do not cause interference to the  primary nodes, we can use \cite[Theorem 3.6]{MeesterRoy1996} to conclude that  the number of unbounded connected components in the primary network, i.e., $N_p$, is at most one a.s.\
We next prove \textit{by contradiction} that the number of unbounded connected components in the secondary network is at most one a.s.\
Note that the proof generalizes the proof in \cite[Proposition 3.3]{MeesterRoy1996} and the proof of \cite[Lemma 2]{RenZhao2011}.

Assume towards contradiction that the number of unbounded components in the secondary network is greater than one, i.e., $N_s\geq 2$.
Suppose that $N_s$ is a finite number such that $N_s\geq 2$. By Proposition \ref{theorem:number_unbounded_constant} it suffices to contradict this assumption by proving that with a positive probability all $N_s$ unbounded connected components can be linked together (by adding secondary nodes and deleting primary nodes) without affecting the number of unbounded connected components $N_p$, in the primary network a.s.\

\begin{figure}
\centering
\includegraphics[scale=1]{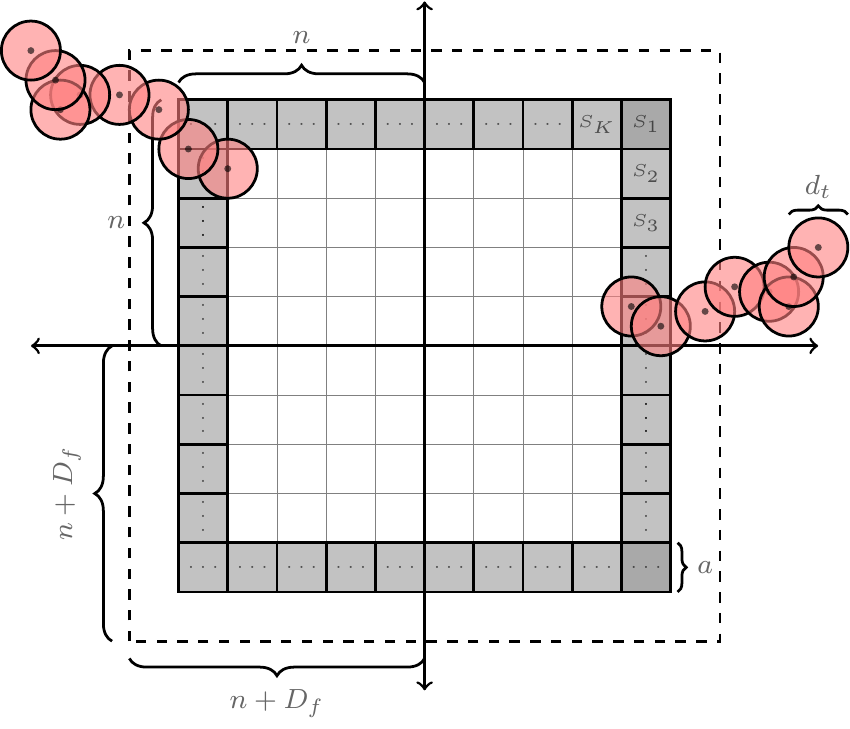}
\caption{Linking two connected (occupied) components of the secondary network. }
\label{fig_one_comp}
\end{figure}

By assumption $N_s$ is finite, therefore, there exists $n\in\mathbb{N} $ such that the box $B=[-n,n]^2$ includes at least one secondary node from each unbounded connected component. For each $A\subset \mathbb{R}^2$, let $C_s[A]$ denote the occupied region formed by the secondary nodes in $A$, that is, $C_s[A]=\cup_{z_i\in A}O(z_i,d_t/2)$. For a box $B$ and some $\epsilon>0$ we define the event $A(B,\epsilon)$ by:
\begin{flalign}
&A(B,\epsilon)\triangleq \{\tilde{\omega}\in\tilde{\Omega}:d(U,B)\leq d_t/2-\epsilon \text{ for any }\nonumber\\
&\hspace{3.6cm} \text{unbounded connected component } U \text{ in } C_s[B^{c}]\}.
\end{flalign}
This event includes all the events in which for each occupied component in the secondary network  outside the box $B$ there exists a secondary node which is within a distance of $\frac{d_t}{2}-\epsilon$ from  $B$.

Partition the box $B$ into squared cells with side length $a=a(B,\epsilon)>0$, and let $\mathcal{S}_a=\{S_1,\ldots,S_K\}$ be the collection of all cells that are adjacent to the boundary of the box $B$. For every box $B$ and $\epsilon>0$ we can find $a=a(B,\epsilon)>0$ and $\eta=\eta(a)>0$ such that for any node $z$ in the secondary network such that $z\notin B$ and $d(z,B)\leq \frac{d_t}{2} -\epsilon$ there exists a square $S\in\mathcal{S}_a$ such that $\sup_{z_B\in S}d(z,z_B)\leq d_t-\eta$. It follows that if we place a secondary node in each of the boundary cells $\mathcal{S}_a$ and if there are no primary nodes in  $\overline{B}=[-n-D_f,n+D_f]^2$, then every unbounded connected component $U$ in $C_s[B^c]$ such that
$d(U,B)\leq \frac{d_t}{2}-\epsilon$ is connected to a node in $\mathcal{S}_a$ as in Fig.\ \ref{fig_one_comp}.

Let $A(a,\eta)$ be the event where there exists at least one secondary node in each square in $\mathcal{S}_a$, and let
$A_{p}(\overline{B})$ be the event where there are no primary nodes in the larger box $B$.
We get that\begin{flalign}
&\Pr(A(B,\epsilon)\cap A(a,\eta)\cap A_{p}(\overline{B})\cap \{N_p=k\})\nonumber\\
&=\Pr(A(B,\epsilon))\Pr( A(a,\eta))\Pr( A_{p}(\overline{B})|A(B,\epsilon)))\nonumber\\
&\quad\cdot\Pr(N_p=k|A(B,\epsilon)\cap A(a,\eta)\cap A(\overline{B})).
\end{flalign}
As in the proof in \cite[Proposition 3.3]{MeesterRoy1996} and the proof in \cite[Lemma 2]{RenZhao2011}, there exists a box $B$, and constants $\epsilon,a,\eta>0$ such that
$\Pr(A(B,\epsilon)),\Pr( A(a,\eta))>0$.  Additionally,
\begin{flalign}
&Pr( A_{p}(\overline{B})|A(B,\epsilon)))\geq Pr( A_{p}(\overline{B}))>0
\end{flalign}
where  the rightmost inequality holds since the box $B$ is bounded.

Since
$\Pr(A(B,\epsilon))\Pr( A(a,\eta))\Pr( A(\overline{B})|A(B,\epsilon)))>0$, by the law of total probability if $Pr(N_p=k)=1$ then $\Pr(N_p=k|A(B,\epsilon)\cap A(a,\eta)\cap A(\overline{B}))>0$, and if $Pr(N_p=k)=0$, then $\Pr(N_p=k|A(B,\epsilon)\cap A(a,\eta)\cap A(\overline{B}))=0$. Therefore by Proposition \ref{theorem:number_unbounded_constant} there are either $0$ or $1$  unbounded connected components in the primary network and either $0,1$ or $\infty$ unbounded connected components in the secondary network.

It remains to be proven that there cannot be an infinite number of unbounded connected components in the secondary network. By implementing \cite[Lemma 3.2]{MeesterRoy1996} as in the proof in \cite[Theorem 3.6]{MeesterRoy1996} it follows that there cannot be an infinite number of unbounded connected components in the secondary networks. We note that   equality (3.59) in \cite[p.\ 68]{MeesterRoy1996} is replaced by inequality since not all the secondary nodes are active. Other than that, all the steps of the proof by contradiction  which appear in (\cite[p.\ 66-68]{MeesterRoy1996})
hold for the secondary network as well.
\end{proof}

\begin{theorem}\label{theorem:existence_connected_comp}
Let $D_t,d_t>0$ be given. For every $\lambda_p>\lambda_c(D_t)$ and $\lambda_s>\lambda_c(d_t)$ there exists $D_f>0$ such that there is an unbounded connected component in the primary network and also an unbounded connected component in the secondary network.
\end{theorem}

\begin{proof}

First, by the definition of our model if $\lambda_p>\lambda_c(D_t)$ then the primary network includes an unbounded connected component a.s. We proceed to consider the connectivity of the secondary network. For the sake of this proof we discretize the continuous model onto a bond percolation model $\mathcal{L}$ in $\mathbb{R}^2$ such that if there is percolation in the secondary networks in the discrete model it is imperative there is an unbounded connected component in the continuous model.

\begin{figure}
\centering
\includegraphics[scale=0.7]{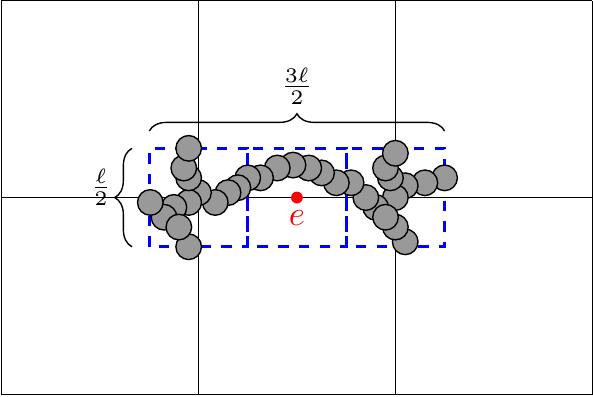}
\caption{Discretization of the continuous model. The edge $e$ is open since the three crossings exist.}
\label{fig_crossing}
\end{figure}

 Set $\ell\in\mathbb{R}^+$ such that $\ell>2(2D_t+d_{t})$ and place the vertices of the planar graph in $(\ell k,\ell m)$ where $k,m\in\mathbb{Z}$. Each vertex $(\ell k,\ell m)$  is connected to the vertices $(\ell (k\pm1),\ell m)$, $(\ell k,\ell (m\pm1))$. The set of all edges is denoted by $E$. Denote the middle point of the edge $e\in E$ by $(x_e,y_e)$. An edge $e\in E$ is said to be open if the following conditions hold:
\begin{enumerate}
\item
\begin{enumerate}
\item There is an L-R occupied crossing of secondary nodes in the rectangle $\left[x_e-\frac{3\ell}{4},x_e+\frac{3\ell}{4}\right]\times \left[y_e-\frac{\ell}{4},y_e+\frac{\ell}{4}\right]$ in \index{This set of conditions} $\mathcal{B}(X_s,d_t/2,\lambda_s)$.
\item There are two T-B occupied crossings of secondary nodes, one crossing in the rectangle
$\left[x_e+\frac{\ell}{4},x_e+\frac{3\ell}{4}\right]\times \left[y_e-\frac{\ell}{4},y_e+\frac{\ell}{4}\right]$
and the other in the box
$\left[x_e-\frac{\ell}{4},x_e-\frac{3\ell}{4}\right]\times \left[y_e-\frac{\ell}{4},y_e+\frac{\ell}{4}\right]$
 in $\mathcal{B}(X_s,d_t/2,\lambda_s)$.
\end{enumerate}
\item There is no primary node within a distance of $D_f$ from any secondary node which composes one of the three crossings.
\end{enumerate}

 For each edge $e\in E$ we define two (dependent) binary random variables $A_e$ and $B_e$. We set $A_e=1$ if condition (1) holds and $A_e=0$ otherwise. Similarly, $B_e=1$ if $A_e=1$ and condition (2) holds and $B_e=0$ otherwise. The state of an edge $e$ is denoted by $C_e=A_eB_e$, where $'1'$ stands for an open edge and $'0'$ a closed one.
Note that the states of the edges are dependent; however, for $\ell>2(2D_t+d_{t})$ the state of an edge  only depends  on the states of its six neighboring edges.

 Let  $(0,0)$ be the origin vertex. Denote by $C(0)$ the set of active secondary vertices in $\mathcal{L}$ which are connected to the origin of $\mathcal{L}$ by open paths. The number of vertices in $C(0)$ is denoted by $|C(0)|$.
By  Proposition \ref{theorem:ergodicity} and Theorem \ref{theorem:one_unbounded_constant} it suffices to prove that $\Pr(|C(0)|=\infty)>0$ to prove that there is an unbounded connected component in the secondary network a.s.

We prove that  $\Pr(|C(0)|=\infty)>0$ by using ``Peierls argument" \cite[pp. 16]{Grimmett1999}.
Let $\mathcal{E}_e$ be the event such that the edge $e$ is closed, that is,
\begin{flalign}
\mathcal{E}_e(\ell,D_f)=\{A_e=0\}\cup\{B_e=0\}.
\end{flalign}
Therefore,\begin{flalign}
q(\ell,D_f)\triangleq\Pr(\mathcal{E}_e(\ell,D_f))\leq \Pr(A_e=0)+\Pr(B_e=0).
\end{flalign}

Let $\mathcal{L}'$ be the dual graph of $\mathcal{L}$.
Then,\begin{flalign}
&\Pr(|C(0)|=\infty)=\Pr(\exists \text{ an infinite open path in }  \mathcal{L})\nonumber\\
&=1-\Pr(\exists \text{ a closed circuit in }  \mathcal{L}' \text{ which contains the origin in its interior})
\end{flalign}
Denote by $\rho(n)$ the number of circuits of length $n$ which contain the origin in their interior. Then, $\rho(n)$ is upper bounded by \cite[pp. 15-18]{Grimmett1999}
\begin{flalign}
\rho(n)\leq 4n3^{n-2}.
\end{flalign}
Let $M(n)$ be the number of closed circuits of length $n$ which contain the origin in their interior, and let $q=q(\ell,D_f)$. Then,
\begin{flalign}
&\Pr(\exists \text{ a closed circuit in }  \mathcal{L}' \text{ which contains the origin in its interior})\nonumber\\
&=\Pr(M(n)\geq 1 \text{ for some } n)\leq\sum_{n=1}^{\infty}\Pr(M(n)\geq 1)\leq\sum_{n=1}^{\infty}\rho(n)q^{n/4}\nonumber\\
&\leq\sum_{n=1}^{\infty}4n3^{n-2}q^{n/4}=\frac{4q^{1/4}}{3}\sum_{n=1}^{\infty}n\left(3q^{1/4}\right)^{n-1}=\frac{4q^{1/4}}{3(1-3q^{1/4})^2}.
\end{flalign}
It follows that $\Pr(|C(0)|=\infty)>0$ if $q<\left(\frac{11-2\sqrt{10}}{27}\right)^4$.

We next show that there are $\ell$ and $D_f$ such that $q(\ell,D_f)<\left(\frac{11-2\sqrt{10}}{27}\right)^4$.
Let $\alpha$ be the (random) area in which primary nodes interfere with communication of secondary nodes in the L-R or one of the T-B crossings. By definition it follows that $0\leq\alpha\leq \left(\frac{3\ell}{2}+2D_f+d_t\right)\left(\frac{\ell}{2}+2D_f+d_t\right)\triangleq S_{p,\max}(\ell)$.
Moreover, since not all the nodes that form the L-R and T-B crossings are essential for the existence of these crossings, it follows that
\begin{flalign}
\Pr(B_e=1|A_e=1,\alpha)
\geq e^{-\lambda_p\alpha}.
\end{flalign}
Therefore,
 \begin{flalign}
\Pr(B_e=0|A_e=1)&=1-\Pr(B_e=1|A_e=1)\nonumber\\
&\leq 1-\int_0^{S_{p,\max}(\ell)}f(\alpha|A_e=1)e^{-\lambda_p\alpha}d\alpha.
\end{flalign}
The function $e^{-\lambda_p\alpha}$ is convex with respect to $\alpha$, therefore, by Jensen's inequality:
\begin{flalign}
\Pr(B_e=0|A_e=1) \leq 1-e^{-\lambda_p E(\alpha|A_e=1)}.
\end{flalign}
Next, we upper bound $E(\alpha|A_e=1)$.
Let $K$ be the (random) number of secondary nodes in the box $\left[x_e-\frac{3\ell+2d_{t}}{4},x_e+\frac{3\ell+2d_{t}}{4}\right]\times \left[y_e-\frac{\ell+2d_{t}}{4},y_e+\frac{\ell+2d_{t}}{4}\right]$, and
denote by $\boldsymbol z_{k}=(z_{1},\dots,z_{k})$  the vector of the (random) positions of these $K$ secondary nodes.
Also, let $I(\boldsymbol z_{k})$ be the area of the region $\bigcup_{i=1}^k O(z_{i},D_f)$. Given that $K=k,$ one has that $\alpha\leq I(\boldsymbol z_{k})$ since the  region $\bigcup_{i=1}^k O(z_{i},D_f)$ may include secondary nodes which are not part of the crossings. Also, by definition $I(\boldsymbol z_{k})$ is upper bounded by $k\pi D_f^2$. We proceed to bounding $E(\alpha|A_e=1)$,
\begin{flalign}
&E(\alpha|A_e=1)=\int_0^{S_{p,\max}(\ell)}\alpha f(\alpha|A_e=1)d\alpha\nonumber\\
&=\int_0^{S_{p,\max}(\ell)}\alpha \sum_{k=1}^{\infty} \Pr(K=k|A_e=1) f(\alpha|A_e=1,K=k)d\alpha\nonumber\\
&= \sum_{k=1}^{\infty} \Pr(K=k|A_e=1) \int_0^{S_{p,\max}(\ell)}\alpha f(\alpha|A_e=1,K=k)d\alpha\nonumber\\
&\stackrel{(a)}{\leq}\sum_{k=1}^{\infty} \Pr(K=k|A_e=1) \int_0^{S_{p,\max}(\ell)}k\pi D_f^2\cdot f(\alpha|A_e=1,K=k)d\alpha\nonumber\\
&= \sum_{k=1}^{\infty}\Pr(K=k|A_e=1)k\pi D_f^2\nonumber\\
&= \frac{1}{\Pr(A_e=1)}\sum_{k=1}^{\infty}\Pr(K=k,A_e=1)k\pi D_f^2\nonumber\\
&\leq\frac{1}{\Pr(A_e=1)}\sum_{k=1}^{\infty}\Pr(K=k)k\pi D_f^2\nonumber\\
&=\pi D_f^2\frac{\sum_{k=1}^{\infty}\Pr(K=k)k}{\Pr(A_e=1)}\nonumber\\
&\stackrel{(b)}{=} \pi D_f^2\frac{ \lambda_s }{\Pr(A_e=1)}\cdot\frac{(3\ell+d_t)(\ell+d_t)}{4},
\end{flalign}
where $(a)$ follows since given that $K=k$, the interference area $\alpha$ is upper bounded by  $k\pi D_f^2$  and $(b)$ follows since $K$ is a Poisson random variable with density $\lambda_s\frac{(3\ell+d_t)(\ell+d_t)}{4}$.
Thus,
\begin{flalign}
\Pr(B_e=0|A_e=1) \leq 1-e^{-\pi\lambda_p D_f^2\frac{ \lambda_s }{\Pr(A_e=1)}\cdot\frac{(3\ell+d_t)(\ell+d_t)}{4}}.
\end{flalign}
By the law of total probability,
\begin{flalign}
&q(\ell,D_f)=\Pr(\mathcal{E}_e(\ell,D_f))\leq \Pr(A_e=0)+\Pr(B_e=0)\nonumber\\
&\leq \Pr(A_e=0)+\Pr(A_e=0)\Pr(B_e=0|A_e=0)\nonumber\\
&\qquad+\Pr(A_e=1)\Pr(B_e=0|A_e=1)\nonumber\\
&\leq 2\Pr(A_e=0)+\Pr(B_e=0|A_e=1)\nonumber\\
&\leq 2\Pr(A_e=0)+1-e^{-\pi\lambda_p D_f^2\frac{ \lambda_s }{\Pr(A_e=1)}\cdot\frac{(3\ell+d_t)(\ell+d_t)}{4}}.
\end{flalign}
Since $\lambda_s>\lambda_c\left(\frac{d_t}2\right)$, $\Pr(A_e=0)$ vanishes as $\ell$ tends to infinity. Therefore, for every $\epsilon>0$ there exists $\ell\in\mathbb{R}$ such that
$\Pr(A_e=0)<\epsilon/3$. Furthermore, one can choose $D_f>0$ such that $1-e^{-\pi\lambda_p D_f^2\frac{\lambda_s }{\Pr(A_e=1)}\cdot\frac{(3\ell+d_t)(\ell+d_t)}{4}}<\epsilon/3$.
Consequently, there are $\ell,D_f>0$ such that $q<\left(\frac{11-2\sqrt{10}}{27}\right)^4$.
\end{proof}

\section{Necessary Conditions for Simultaneous Connectivity }\label{sec:nec_cond}

In this section we state the necessary conditions for simultaneous percolation in both the primary and secondary networks. These two conditions are found by implementing two different methods.
The first condition, stated in Theorem \ref{theorem:necessary_cond1}, is found by considering  the fact
that there cannot exist  both an unbounded vacant component and an occupied component in a Gilbert disk (Boolean) model a.s.\ The second condition, stated in Theorem \ref{theorem:necessary_cond2}, is found by discretizing the continuous model onto a site percolation model and bounding the connected component a.s.\
\begin{theorem}\label{theorem:necessary_cond1}
Suppose that $2D_f > d_t$, then
\begin{flalign}\label{eq:necessary_cond1}
&\lambda_s > d_t^{-2}\lambda_c(1)\nonumber\\
&\lambda_p > D_t^{-2}\lambda_c(1)\nonumber\\
&\lambda_p <\left(4D_f^2-d_t^2\right)^{-1}\lambda_c(1)
\end{flalign}
are necessary  conditions for  simultaneous percolation in both networks.
\end{theorem}
\begin{proof}

\begin{figure}
\centering
\includegraphics[scale=1]{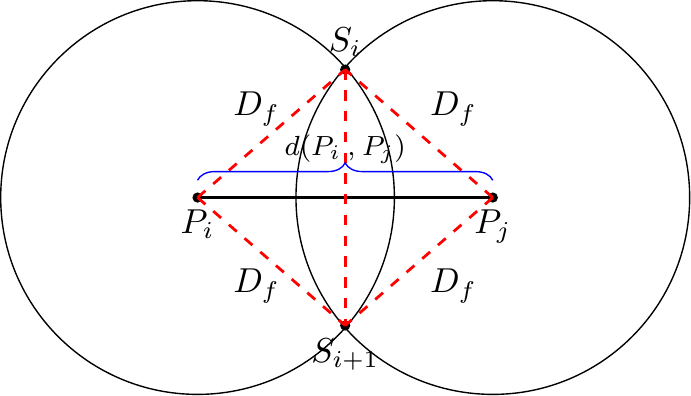}
\caption{The guard zone of two adjacent primary nodes.}
\label{fig_necessary_cond}
\end{figure}
Since interference between the networks can only reduce the connectivity, the first two conditions of Eq.\ \ref{eq:necessary_cond1} follow from the fact that there is no percolation in either of the networks unless the density of the nodes is greater than the critical density, assuming no interference.
The proof of the third condition of Eq.\ \ref{eq:necessary_cond1} is inspired by the proof of Theorem 2.2 in \cite{RenZhao2011}. In this proof we find a PPP $X_p$ with density and radius $\rho$ such that if there is no unbounded vacancy region generated by $(X_p,\rho,\lambda_p)$  there is no unbounded connected component in the secondary network. A simple example as in \cite[Fig.\ 13]{RenZhao2011} shows that the existence of a vacancy region in the PPP $\mathcal{}(X_p,D_f,\lambda_p)$ is not a necessary condition. That is, there
may exist unbounded connected components in both the primary and secondary networks even if there is no unbounded vacancy region in $(X_p,D_f,\lambda_p)$.

Denote by $V_{\rho}$ the largest vacancy component of the PPP $(X_p,\rho,\lambda_p)$. 
Our goal is to find $\rho$ such that if $V_{\rho}$ is bounded  there cannot be $L-R$ or $T-B$ crossings in the secondary network.

Suppose that two disks of $(X_p,\rho,\lambda_p)$ intersect, and let  $P_i$ and $P_j$ be the centers of these disks. The shortest crossing option for secondary nodes  is through the points $S_i$ and $S_{i+1}$. This distance is the shortest when the distance between $P_i$ and $P_j$ is $2\rho$. However, suppose that $\rho < \sqrt{D_f^2-\frac{d_t^2}{4}}$ then from Fig.\ \ref{fig_necessary_cond}, one can see that if there is no unbounded vacant component in $(X_p,\rho,\lambda_p)$, an unbounded component does not exist in the secondary network a.s.\

Therefore,
\begin{flalign}
\lambda_{p}<\lambda_c\left(2\sqrt{D_f^2-\frac{d_t^2}{4}}\right) =\frac{\lambda_c(1)}{4D_f^2-d_t^2}.
\end{flalign}
\end{proof}

Another set of conditions for  simultaneous connectivity is obtained by discretization onto a site percolation in which each site has eight neighbors. This set of conditions explores the relationship between the    densities of the primary secondary networks under the simultaneous percolation regime.

\begin{theorem}\label{theorem:necessary_cond2}
Let $n_p\triangleq\Bigl\lceil\frac{\sqrt{2}d_t}{D_f}\Bigr\rceil^2$.
Denote by $p_8$ the critical probability of  site percolation with eight neighbors. Then
\begin{flalign}\label{eq:necessary_cond2}
&\lambda_s > d_t^{-2}\lambda_c(1),\nonumber\\
&\lambda_p > D_t^{-2}\lambda_c(1),\nonumber\\
&\lambda_p<-\frac{n_p}{d_t^2}\ln\left(1-\left(1-e^{-\lambda_s d_t^2}\right)^{-1/n_p}\left(1-e^{-\lambda_sd_t^2}-p_8\right)^{1/n_p}\right)
\end{flalign}
are necessary conditions for simultaneous percolation in both networks.
\end{theorem}

\begin{proof}
  As previously, the first two conditions of Eq. (\ref{eq:necessary_cond2})
are rudimentary conditions for  the connectivity of the primary and the secondary  networks. To establish the third condition of Eq.\ (\ref{eq:necessary_cond2}) we discretize the continuous model onto a discrete model. We choose  a model in which non-occurrence of a percolation dictates that there cannot be an unbounded connected  component in the secondary network.

We discretize the continuous model onto an eight neighbors site model in the following manner.
Partition $\mathbb{R}^2$ into squares of side length $d_t$, i.e., the squares
$s(m,k)\triangleq[-\frac{d_t}{2}+m\cdot d_t,\frac{d_t}{2}+m\cdot d_t]\times [-\frac{d_t}{2}+k\cdot d_t,\frac{d_t}{2}+k\cdot d_t]$ where $k,m\in\mathbb{Z}$.
Partition each square into $n_p^2=\Bigl\lceil\frac{d_t}{D_f/\sqrt{2}}\Bigr\rceil ^2$ sub-squares, each of side length $\frac{d_t}{n_p}$.
We say that a square $s(m,k)$ is closed if it does not contain any secondary nodes or if every sub-square of $s(m,k)$ contains at least one primary user.
The probability that a square is closed is
\begin{flalign}
\left(1-e^{-\lambda_s d_t^2}\right)\left(1-e^{-\lambda_p\frac{d_t^2}{n_p}}\right)^{n_p}+e^{-\lambda_s d_t^2}.
\end{flalign}

Let $p_c$ be the critical probability of a site percolation model. By discrete percolation  if the probability for a site to be closed is greater than $1-p_c$, then every connected  component is finite a.s.\ For the eight neighbors site model this implies that when
\begin{flalign}
\left(1-e^{-\lambda_s d_t^2}\right)\left(1-e^{-\lambda_p\frac{d_t^2}{n_p}}\right)^{n_p}+e^{-\lambda_s d_t^2}>1-p_8.
\end{flalign}
every connected component in the secondary network is bounded a.s.\
Equivalently, an unbounded component may exist in the secondary network only if
\begin{flalign}
\left(1-e^{-\lambda_s d_t^2}\right)\left(1-e^{-\lambda_p\frac{d_t^2}{n_p^2}}\right)^{n_p}+e^{-\lambda_s d_t^2}<1-p_8.
\end{flalign}

Further algebra yields the following condition
\begin{flalign}
\lambda_p<-\frac{n_p}{d_t^2}\ln\left(1-\left(1-e^{-\lambda_s d_t^2}\right)^{-1/n_p}\left(1-e^{-\lambda_sd_t^2}-p_8\right)^{1/n_p}\right).
\end{flalign}
\end{proof}

By applying the lower bound $\frac{1}{3}\leq p_8$ (see \cite[Chapter 2.2]{FranceschettiMeester2007}) to Theorem \ref{theorem:necessary_cond2} we obtain the following corollary.
\begin{corollary}\label{corollary:necessary_cond2}
Let $n_p\triangleq\Bigl\lceil\frac{\sqrt{2}d_t}{D_f}\Bigr\rceil^2$.
 Then
\begin{flalign}\label{eq:cor_necessary_cond2}
&\lambda_s > d_t^{-2}\lambda_c(1),\nonumber\\
&\lambda_p > D_t^{-2}\lambda_c(1),\nonumber\\
&\lambda_p<-\frac{n_p}{d_t^2}\ln\left(1-\left(\frac{2}{3}-e^{-\lambda_sd_t^2}\right)^{1/n_p}\right).
\end{flalign}
are necessary conditions for simultaneous percolation in both networks.
\end{corollary}

\section{Sufficient Conditions for Simultaneous Connectivity}\label{sec:suf_cond}

In this section we present the sufficient conditions for the existence of both primary and secondary connected unbounded components.
We find these conditions by discretizing the continuous model onto a dependent site percolation model \cite{FranceschettiMeester2007,Grimmett1999}. Our objective is to achieve  discretization in such a way that if there exists an unbounded connected occupied component  in the discrete site percolation,  an unbounded connected component exists in the continuous model as well.
\begin{figure}
\centering
\includegraphics[scale=0.7]{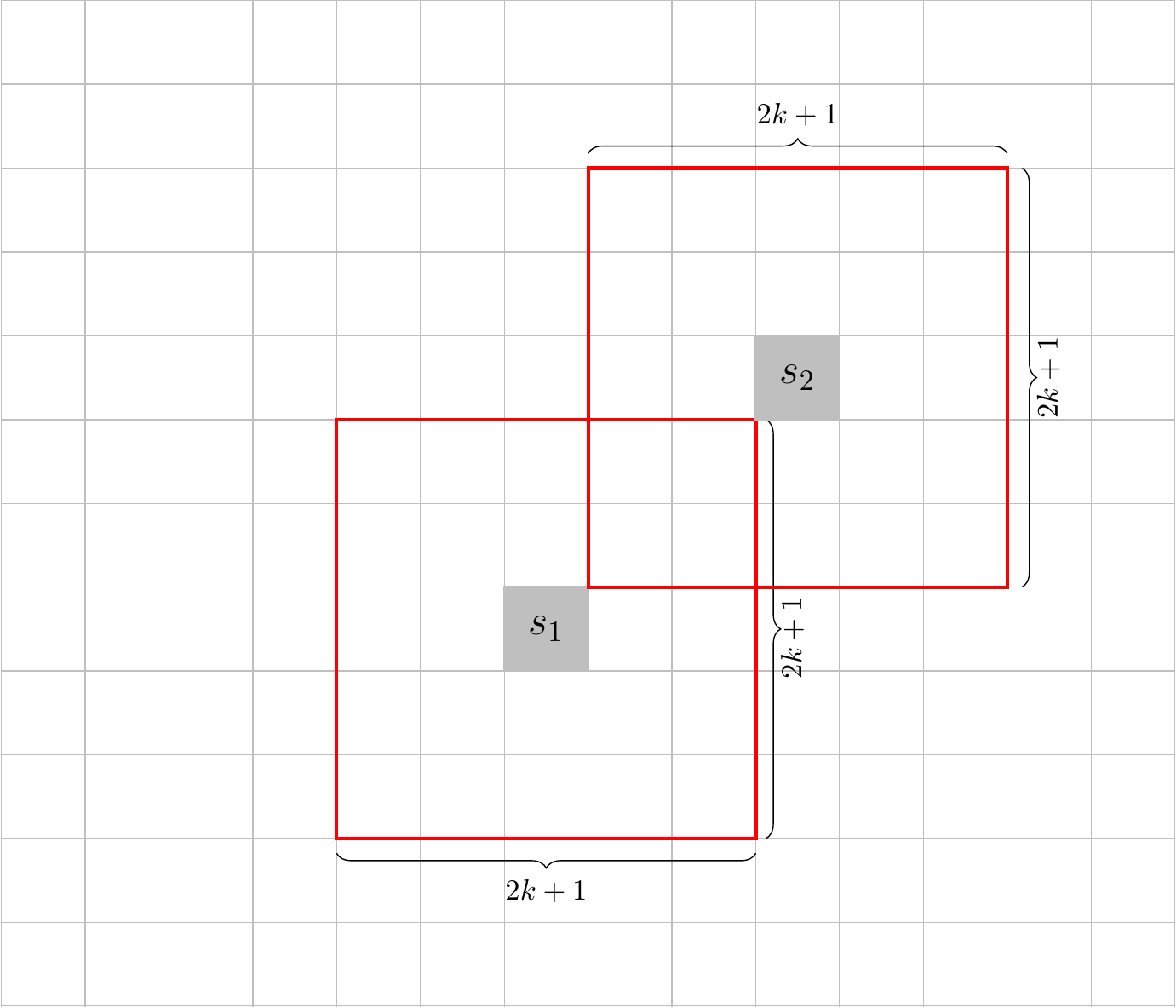}
\caption{A $k$-dependent site percolation. A red box around a site encloses its dependent sites. The sites $s_1$ and $s_2$ are independent.}
\label{fig_dependency}
\end{figure}
A dependent site percolation is a site percolation in which the state of a site may depend on the states of other sites. If the state of a site only depends  on the states of the sites that are separated by a path of minimum length $k<\infty$ we say that the model is $k$-dependent (see Fig.\ \ref{fig_dependency}).

Let $p$ be the marginal probability for a site to be open in a stationary dependent site model. By \cite[Theorem 2.3.1]{FranceschettiMeester2007}, there exists $p(k)$ such that for $p>p(k)$ there exists an unbounded occupied component a.s.\
It follows that there exists $p_4(k)$ such that  for all  $p>p_4(k) $ there exists an unbounded occupied component in the  $k$-dependent site percolation models with four neighbors.
Similarly, there exists $p_{8}(k)$ such that for all  $p>p_{8}(k) $ there exists an unbounded occupied component in the  $k$-dependent site percolation models with eight neighbors.

Note that since the exact values of $p_4(k)$ and $p_8(k)$ are not known we consider both the four and eight site percolation models in Theorem \ref{theorem:sufficient_cond1}. Nonetheless,  Corollary \ref{corollary:sufficient_cond1} considers upper bounds   for which the eight site model yields better results.

Let $S(t,D_f)$ stand for the (random) intersected area between two disks of radius $D_f$ with centers at distance\footnote{
$S(d,r)=-\frac{1}{2}d\sqrt{4r^2-d^2}-2r^2\arctan\left(\frac{d}{\sqrt{4r^2-d^2}}\right)+\pi r^2$.} $t$. Further, let
\begin{flalign}
f_{T,\ell}(t)=&\mathbbm{1}_{\{t\in[0,\ell]\}}\frac{4t}{\ell^4}\left(
\frac{\pi\ell^2}{2} -2\ell t+\frac{t^2}{2}\right)\nonumber\\
&\quad+\mathbbm{1}_{\{t\in[\ell,\sqrt{2}\ell]\}}\frac{4t}{\ell^4}\left(\ell^2\arcsin\left(\frac{2\ell^2-t^2}{t^2}\right)
+2\ell\sqrt{t^2-\ell^2} -\ell^2-\frac{t^2}{2}\right)
\end{flalign}
be the probability density function of the random distance $t$  between  the two centers, each generated independently from the box of side length $\ell$, which is derived in (\ref{density_function}).

Denote, $\ell_4=\frac{D_t}{\sqrt{5}}$, $\ell_8=\frac{D_t}{2\sqrt{2}}$,
\begin{flalign}
&\tilde{p}_4(D_f,d_t,\lambda_p,\lambda_s)\triangleq \lambda_s \frac{d_t^2}{5}e^{-\lambda_s\frac{d_t^2}{5}}e^{-\lambda_p\pi D_f^2}\nonumber\\
&\quad+\left(1-\lambda_s \frac{d_t^2}{5}e^{-\lambda_s\frac{d_t^2}{5}}-e^{-\lambda_s\frac{d_t^2}{5}}\right)e^{-\lambda_p\pi D_f^2}\left[2-\int_{-\infty}^{\infty} f_{T,\ell_{4}}(t) e^{\lambda_p[S(t,D_f)-\pi D_f^2]}dt\right],
\end{flalign}
and
\begin{flalign}
&\tilde{p}_8(D_f,d_t,\lambda _{p},\lambda_s)\triangleq \lambda_s \frac{d_t^2}{8}e^{-\lambda_s\frac{d_t^2}{8}}e^{-\lambda_p\pi D_f^2}\nonumber\\
&\quad+\left(1-\lambda_s \frac{d_t^2}{8}e^{-\lambda_s\frac{d_t^2}{8}}-e^{-\lambda_s\frac{d_t^2}{8}}\right)e^{-\lambda_p\pi D_f^2}\left[2-\int_{-\infty}^{\infty} f_{T,\ell_{8}}(t) e^{\lambda_p[S(t,D_f)-\pi D_f^2]}dt\right],
\end{flalign}

\begin{theorem}\label{theorem:sufficient_cond1}
Let $k_{4}=\left\lceil\frac{2\sqrt{5}D_f}{d_t}\right\rceil$
 and
$k_{8}=\left\lceil\frac{4\sqrt{2}D_f}{d_t}\right\rceil$. There exist  unbounded connected components simultaneously  in both the primary and secondary networks if the following conditions hold
\begin{flalign}\label{eq:sufficient_cond1}
&\lambda_p> D_t^{-2}\lambda_c(1)\nonumber\\
&\max\left\{\frac{\tilde{p}_4(D_f,d_t,\lambda_p,\lambda_s)}{p_4(k_{4})},\frac{\tilde{p}_8(D_f,d_t,\lambda_p,\lambda_s)}{p_8(k_{8})}\right\}> 1.
\end{flalign}
\end{theorem}

\begin{proof}

We  discretize the continuous model onto two site models - a four neighbors site model and an eight neighbors site model. In both of the models  $(k\ell,m\ell),\quad k,m\in\mathbb{Z}$ are the centers of the boxes with side length $\ell\in\mathbb{R}^+$. We say that a box is occupied if there is a secondary user in the box that sees a (symmetric) spectrum opportunity.

We define a site's neighbors, in the four neighbors model, to be all sites that have a common side with the site. We define a site's neighbors, in the eight neighbors model, to be all sites that have a common side point with the site.
 In both models, a site is said to be occupied if there is at least one secondary user which sees an opportunity. In order to limit dependency and also to allow each secondary user in a box to communicate with its neighbor, we set   $\ell$ in the four neighbors site model such that $\sqrt{5}\ell=D_t$, that is $\ell=\frac{D_t}{\sqrt{5}}$. In this scenario the dependency is of
 \begin{flalign}
 k_{4}=\left\lceil\frac{2D_f}{\ell}\right\rceil=\left\lceil\frac{2\sqrt{5}D_f}{d_t}\right\rceil
 \end{flalign}
 sites. In the eight neighbors site model we set $\ell$ such that $\ell=\frac{D_t}{2\sqrt{2}}$ and
 \begin{flalign}
 k_{8}=\left\lceil\frac{2D_f}{\ell}\right\rceil=\left\lceil\frac{4\sqrt{2}D_f}{d_t}\right\rceil
 \end{flalign}
 sites.

Let $p$ be the (marginal) probability that there exists at least one secondary user in the box  $B=[0,\ell]^{2}$ which sees a communication opportunity.
By \cite[Theorem 2.3.1]{FranceschettiMeester2007} there exists a probability $p(k)$, for each of the models, such that for every $p>p(k)$ there is an unbounded occupied component in the discrete model a.s.\ We denote these probabilities by $p_4(k)$ and $p_8(k)$.

Denote the event that there are $m$ secondary nodes in the box $B$ by $E_m$. The $m$ secondary nodes in the box are denoted by $U_i$, $i\in\{1,\ldots,m\}$ where the index $i$ is associated with a user arbitrarily. Further, let $O_i$ be the event in which the secondary user $U_i$ sees an opportunity.  By definition,
\begin{flalign}
p&=\sum_{m=1}^{\infty}\Pr(E_m)\Pr\left(\cup_{i=1}^m O_i\right)\nonumber\\
&\geq \Pr(E_1)\Pr\left( O_1\right)+\sum_{m=2}^{\infty}\Pr(E_m)\Pr\left(O_1\cup O_2\right)\nonumber\\
&= \lambda_s \ell^{2}e^{-\lambda_s\ell^{2}}e^{-\lambda_p\pi D_f^2}+\sum_{m=2}^{\infty}\Pr(E_m)\left[\Pr(O_1)+\Pr(O_2)-\Pr(O_1\cap O_2)\right].
\end{flalign}
Let $T$ be the (random) distance between the two secondary nodes in the box $B$. The probability density function $f_{T,\ell}(t)$ is derived in (\ref{density_function}).
We also denote by $S(t,D_f)$ the intersection area between two disks of radius $D_f$ with centers at distance $t$. It follows that
\begin{flalign}
&p-\lambda_s \ell^{2}e^{-\lambda_s\ell^{2}}e^{-\lambda_p\pi D_f^2}\geq\sum_{m=2}^{\infty}\Pr(E_m)\left[\Pr(O_1)+\Pr(O_2)-\Pr(O_1\cap O_2)\right]\nonumber\\
&\quad=\left(1-\Pr(E_1)-\Pr(E_0)\right)\left[\Pr(O_1)+\Pr(O_2)-\Pr(O_1\cap O_2)\right]\nonumber\\
&\quad=\left(1-\lambda_s \ell^{2}e^{-\lambda_s\ell^{2}}-e^{-\lambda_s\ell^{2}}\right)\left[2e^{-\lambda_p\pi D_f^2}-\Pr(O_1\cap O_2)\right]\nonumber\\
&\quad= \left(1-\lambda_s \ell^{2}e^{-\lambda_s\ell^{2}}-e^{-\lambda_s\ell^{2}}\right)\left[2e^{-\lambda_p\pi D_f^2}-\int_{-\infty}^{\infty} f_{T,\ell}(t) e^{-\lambda_p( 2\pi D_f^2- S(t,D_f))}dt\right]\nonumber\\
&\quad= \left(1-\lambda_s \ell^{2}e^{-\lambda_s\ell^{2}}-e^{-\lambda_s\ell^{2}}\right)e^{-\lambda_p\pi D_f^2}\left[2-\int_{-\infty}^{\infty} f_{T,\ell}(t) e^{\lambda_p[S(t,D_f)-\pi D_f^2]}dt\right].
\end{flalign}
\end{proof}

A simpler but looser bound can be derived.
\begin{corollary}\label{corollary:sufficient_cond1}
There exists an unbounded connected component in both the primary and secondary networks if the following conditions hold
\begin{flalign}
&\lambda_p> D_t^{-2}\lambda_c(1)\nonumber\\
&\lambda_p<(\pi D_f^2)^{-1}\ln\left(\frac{1-e^{-\lambda_s d_t^2/8}}{1-\left(\frac{1}{3}\right)^{(2k_8+1)^2}}\right).
\end{flalign}
\end{corollary}

\begin{proof}
Let $p$ be the (marginal) probability that there exists at least one secondary user in the box  $B=[0,\ell]^{2}$ which sees a communication opportunity.  By the proof of Theorem \ref{theorem:sufficient_cond1}
  \begin{flalign}
p&=\sum_{m=1}^{\infty}\Pr(E_m)\Pr\left(\cup_{i=1}^m O_i\right)\nonumber\\
&\geq \sum_{m=1}^{\infty}\Pr(E_m)\Pr\left(O_1\right)= e^{-\lambda_p\pi D_f^2}\sum_{m=1}^{\infty}\Pr(E_m)\nonumber\\
&=e^{-\lambda_p\pi D_f^2}\left(1-e^{-\lambda_s\pi \ell^2}\right)=e^{-\lambda_p\pi D_f^2}\left(1-e^{-\lambda_s \frac{d_t^2}{8}}\right).
\end{flalign}

Additionally, by Sections 2.2 and 2.3 in \cite{FranceschettiMeester2007}, $1-p_8(k)< \left(\frac{1}{3}\right)^{(2k_{8}+1)^2}$, it follows that
\begin{flalign}
p_8(k)> 1-\left(\frac{1}{3}\right)^{(2k_{8}+1)^2}.
\end{flalign}
Therefore, a more restrictive requirement is
\begin{flalign}\label{res_eq}
e^{-\lambda_p\pi D_f^2}\left(1-e^{-\lambda_s \frac{d_t^2}{8}}\right) > 1-\left(\frac{1}{3}\right)^{(2k_{8}+1)^2}.
\end{flalign}
Eq.\ (\ref{res_eq}) can be reorganized in the following manner
\begin{flalign}
&e^{-\lambda_p\pi D_f^2} > \frac{1-\left(\frac{1}{3}\right)^{(2k_{8}+1)^2}}{1-e^{-\lambda_s \frac{d_t^2}{8}}}\nonumber\\
&\lambda_p < \frac{1}{\pi D_f^2}\ln\left(\frac{1-e^{-\lambda_s \frac{d_t^2}{8}}}{1-\left(\frac{1}{3}\right)^{(2k_8+1)^2}}\right).
\end{flalign}
\end{proof}

\section{Conclusion}\label{sec:conclusion}
In this paper we presented several properties of the simultaneous connectivity of the heterogeneous network model. We proved that there cannot exist more than one unbounded connected component in each of the networks. Moreover, we presented sufficient  as well as necessary conditions for the simultaneous connectivity of the heterogeneous model. Furthermore, we argued that for each pair of densities greater than the critical density without inter-network interference, there exists a small enough guard zone such that there exist unbounded connected components in each of the networks.
We hope that these results will motivate further discussion on applications and performance of such heterogenous ad-hoc networks.
\appendix
\section{}
In this appendix we derive the density function $f_{T,\ell}(t)$ of the random distance $t$ between two centers of disks (see the proof of Theorem \ref{theorem:sufficient_cond1}).

 Let $X_1,X_2,Y_1,Y_2\sim U[0,\ell]$ be statistically independent.
Define the random variables $T_x$ and $T_y$ in the following manner:
\begin{flalign}
T_x&=X_1-X_2,\nonumber\\
T_y&=Y_1-Y_2.
\end{flalign}
We find the probability density function $T_x$ by the transformation formula.
Let $V=X_2$, then
\begin{flalign}
f_{T_x,V}(t_x,v)=\frac{f_{X_1,X_2}(t_x+v,v)}{1}=f_{X_1}(t_x+v)f_{X_2}(v)
\end{flalign}
By the law of total probability:
\begin{flalign}
f_{T_x}(t_x)&=\int_{-\infty}^{\infty}f_{X_1}(t_x+v)f_{X_2}(v)dv\nonumber\\
&=\frac{1}{\ell^2}\int_{-\infty}^{\infty}\mathbbm{1}_{\{t_x+v\in[0,\ell]\}}\mathbbm{1}_{\{v\in[0,\ell]\}}dv\nonumber\\
&=\frac{\ell-|t_x|}{\ell^2}\cdot \mathbbm{1}_{\{|t_x|\in[0,\ell]\}}
\end{flalign}
It follows that,
\begin{flalign}
f_{|T_x|}(|t_x|)=2\cdot \frac{\ell-|t_x|}{\ell^2}\cdot \mathbbm{1}_{\{|t_x|\in[0,\ell]\}}
\end{flalign}
Similarly,
\begin{flalign}
f_{|T_y|}(|t_y|)=2\cdot \frac{\ell-|t_y|}{\ell^2}\cdot \mathbbm{1}_{\{|t_y|\in[0,\ell]\}}.
\end{flalign}
Define the random variables
\begin{flalign}
&T=\sqrt{|T_x|^2+|T_y|^2}\nonumber\\
&W=|T_y|
\end{flalign}
By the transformation formula
\begin{flalign}
f_{T,W}(t,w)&= f_{|T_x|}\left(\sqrt{t^2-w^2}\right)f_{|T_y|}(w)\left(\frac{2\sqrt{t^2-w^2}}{2t}\right)^{-1}\nonumber\\
&=\frac{4t}{\ell^4}\frac{(\ell-\sqrt{t^2-w^2})(\ell-w)}{\sqrt{t^2-w^2}} \mathbbm{1}_{\{\sqrt{t^2-w^2}\in[0,\ell]\}}\mathbbm{1}_{\{w\in[0,\ell]\}}
\end{flalign}
By the law of total probability
\begin{flalign}\label{density_function}
f_{T,\ell}(t)&=\frac{4t}{\ell^4}\int_{\infty}^{\infty}\frac{(\ell-\sqrt{t^2-w^2})(\ell-w)}{\sqrt{t^2-w^2}} \mathbbm{1}_{\{\sqrt{t^2-w^2}\in[0,\ell]\}}\mathbbm{1}_{\{w\in[0,\ell]\}}dw\nonumber\\
&=\mathbbm{1}_{\{t\in[0,\ell]\}}\frac{4t}{\ell^4}\int_{0}^{t}\left(\frac{\ell^2}{\sqrt{t^2-w^2}}-\frac{\ell w}{\sqrt{t^2-w^2}} -\ell+w\right)dw\nonumber\\
&\quad+\mathbbm{1}_{\{t\in[\ell,\sqrt{2}\ell]\}}\frac{4t}{\ell^4}\int_{\sqrt{t^2-\ell^2}}^{\ell}\left(\frac{\ell^2}{\sqrt{t^2-w^2}}-\frac{\ell w}{\sqrt{t^2-w^2}} -\ell+w\right)dw\nonumber\\
&=\mathbbm{1}_{\{t\in[0,\ell]\}}\frac{4t}{\ell^4}\left.\left(
\ell^2\arcsin\left(\frac{w}{t}\right)+\ell\sqrt{t^2-w^2} -\ell w+\frac{w^2}{2}\right)\right|_{0}^{t}\nonumber\\
&\quad+\mathbbm{1}_{\{t\in[\ell,\sqrt{2}\ell]\}}\frac{4t}{\ell^4}\left.\left(\ell^2\arcsin\left(\frac{w}{t}\right)+\ell\sqrt{t^2-w^2} -\ell w+\frac{w^2}{2}\right)\right|_{\sqrt{t^2-\ell^2}}^{\ell}\nonumber\\
&=\mathbbm{1}_{\{t\in[0,\ell]\}}\frac{4t}{\ell^4}\left[\left(
\ell^2\arcsin\left(1\right) -\ell t+\frac{t^2}{2}\right)-\ell t \right]\nonumber\\
&\quad+\mathbbm{1}_{\{t\in[\ell,\sqrt{2}\ell]\}}\frac{4t}{\ell^4}\left[\left(\ell^2\arcsin\left(\frac{\ell}{t}\right)+\ell\sqrt{t^2-\ell^2} -\ell^2+\frac{\ell^2}{2}\right)\right.\nonumber\\
&\quad\left.-\left(\ell^2\arcsin\left(\frac{\sqrt{t^2-\ell^2}}{t}\right)+\ell^2 -\ell \sqrt{t^2-\ell^2}+\frac{t^2-\ell^2}{2}\right)\right]\nonumber\\
&=\mathbbm{1}_{\{t\in[0,\ell]\}}\frac{4t}{\ell^4}\left(
\frac{\pi\ell^2}{2} -2\ell t+\frac{t^2}{2}\right)\nonumber\\
&\quad+\mathbbm{1}_{\{t\in[\ell,\sqrt{2}\ell]\}}\frac{4t}{\ell^4}\left[\ell^2\left(\arcsin\left(\frac{\ell}{t}\right)-\arcsin\left(\frac{\sqrt{t^2-\ell^2}}{t}\right)\right)
+2\ell\sqrt{t^2-\ell^2} -\ell^2-\frac{t^2}{2}\right]\nonumber\\
&=\mathbbm{1}_{\{t\in[0,\ell]\}}\frac{4t}{\ell^4}\left(
\frac{\pi\ell^2}{2} -2\ell t+\frac{t^2}{2}\right)\nonumber\\
&\quad+\mathbbm{1}_{\{t\in[\ell,\sqrt{2}\ell]\}}\frac{4t}{\ell^4}\left(\ell^2\arcsin\left(\frac{2\ell^2-t^2}{t^2}\right)
+2\ell\sqrt{t^2-\ell^2} -\ell^2-\frac{t^2}{2}\right)\nonumber\\
\end{flalign}

\bibliographystyle{amsplain}
\providecommand{\bysame}{\leavevmode\hbox to3em{\hrulefill}\thinspace}
\providecommand{\MR}{\relax\ifhmode\unskip\space\fi MR }
\providecommand{\MRhref}[2]{%
  \href{http://www.ams.org/mathscinet-getitem?mr=#1}{#2}
}
\providecommand{\href}[2]{#2}

\printindex

\end{document}